\documentclass[a4paper,10pt]{scrartcl}

\usepackage{etex}

\usepackage{fullpage}

\usepackage{lmodern}

\usepackage[usenames,dvipsnames]{xcolor}

\usepackage{amsmath}%
\usepackage{amsfonts}%
\usepackage{amssymb}
\usepackage{graphicx}
\usepackage{esint}
\usepackage{enumitem}
\usepackage{appendix}
\usepackage{multicol}
\usepackage{multirow}
\usepackage{pstricks,pst-plot}
\usepackage{wrapfig}
\usepackage{tikz-cd}
\usetikzlibrary{matrix,arrows,decorations.pathmorphing}
\usepackage{longtable}
\usepackage{eqnarray}
\usepackage{mathtools}

\usepackage[british]{babel}

\usepackage{url}

\usepackage[stable]{footmisc}
\usepackage{footnote}
\makesavenoteenv{align*}

\usepackage{amsthm}

\usepackage{aliascnt}

\newtheorem{theorem}{Theorem}
\theoremstyle{plain}

\newaliascnt{corollary}{theorem}  
\newtheorem{corollary}[corollary]{Corollary}
\aliascntresetthe{corollary}

\newaliascnt{lemma}{theorem}  
\newtheorem{lemma}[lemma]{Lemma}  
\aliascntresetthe{lemma}

\newaliascnt{proposition}{theorem}  
  
\aliascntresetthe{proposition}

\newtheorem{result}[theorem]{Result}
\newtheorem*{result*}{Result}

\theoremstyle{definition}

\newtheorem{definition}[theorem]{Definition}

\theoremstyle{remark}

\newtheorem*{claim*}{Claim}

\newtheorem{remark}[theorem]{Remark}

\numberwithin{equation}{section}

\DeclareMathOperator{\orb}{orb}

\newcommand{\abs}[1]{\left|#1\right|}
\newcommand{\norm}[1]{\left\lVert#1\right\rVert}

\newcommand{\reals}{\ensuremath{\mathbb{R}}}

\newcommand{\complexes}{\ensuremath{\mathbb{C}}}

\usepackage[figureposition=top,tableposition=top]{caption}
\usepackage[position=above]{subfig}

\usepackage[url=false,doi=false,mincrossrefs=1,backend=bibtex]{biblatex}
\usepackage{csquotes}
\newmuskip\pFqskip
\mathchardef\pFcomma=\mathcode`, 

\usepackage{hyperref}

\renewcommand{\leq}{\leqslant}
\renewcommand{\geq}{\geqslant}

\usepackage[affil-it]{authblk}

\title{Symmetric Potentials Beget Symmetric Ground States}
\author{Richard Chapling}
\affil{Department of Applied Mathematics and Theoretical Physics, \\ University of Cambridge, Cambridge, England
}

\begin{document}

\newcommand{\Dim}{d}
\newcommand{\cX}{\mathcal{X}}
\newcommand{\cY}{\mathcal{Y}}

\maketitle

\begin{abstract}
  Using an unusual type of symmetric average, we show that for several common equations involving quite general potentials possessing symmetry, the ground state, if it exists, must also be symmetric.
\end{abstract}

\section{Introduction}
\label{sec:introduction}

It is part of folklore in mathematical physics that in many cases a symmetric potential produces a symmetric ground state. It is our intention here to prove this in a fairly general setting, and for a selection of terms in the equations. We shall also consider only problems in which the boundary conditions are homogeneous and Dirichlet, \( u=0 \) on \(\partial X \); this avoids situations such as those discussed in \cite{Esteban:1991qy,Coti-Zelati:1990fj}, where symmetry breaking is known to occur.

In cases where the potential is decreasing, this is readily demonstrated, as may for example be found in \cite{lieb2001analysis}, Ch. 11 for the Schr\"odinger equation and \cite{2008arXiv0807.4059S}, \cite{Choquard:2007fk} and \cite{Lieb:1977uq} for the Schr\"odinger--Newton equations in various settings; we give a summary of this in the next section.

The procedure that we employ goes by various names in the literature: rearrangement, symmetrisation, and in the case that the symmetry group is a rotation group, ``radialisation''.

We begin by giving definitions related to the spaces we shall work on, focussing on generality, and then give several examples. The third section introduces the terms in the energy functionals that we shall consider. The fourth discusses the standard approach for increasing potentials using symmetric--decreasing rearrangement.

In the fifth section, we introduce our new mean values and discuss its basic properties, most of which are similar to those of the standard \(L^p\)-norms or symmetric--decreasing rearrangement. In subsequent sections, we discuss each term given in the second section, and finally in Section 8 give the proof of the main theorem. The last two sections discuss some further generalisations, to certain relativistic kinetic energies, and other nonlinear terms.

\section{Definitions}
\label{sec:definitions}

In this section we define precisely the various terms we use in this paper.

Throughout, we take 
\begin{itemize}
\item \(X\) a measurable space with a weak differentiation operator \(\nabla\) defined on it.
\item \(G\) a compact group acting on \(X\);
\item \(dg\) an invariant probability measure on \(G\),
\item \(dx\) a \(G\)-invariant measure on \(X\).
\end{itemize}

The operator \( \nabla \) is required to act is a similar way to the weak derivative on \( \reals^n \): it satisfies the \( X \to \complexes \to \complexes \) chain rule \( \nabla f(u(x)) = df(u(x)) \cdot \nabla u(x) \) in the weak sense.

We give a small list of examples of this kind of system, both to illustrate the wide application of this technique, and to enable the reader to fix concepts in their mind:
\begin{enumerate}
\item Let \(X = \reals \), and take \(G \cong Z_2 \), the non-identity element acting by reflection in the origin, \( x \mapsto -x \). Then \( \int_G f(g) \, dg = \frac{1}{2}(f(1)+f(-1)) \), and standard Lebesgue measure \(dx\) is \(G\)-invariant.
\item Let \( X = \complexes \), and take \(G\) to be the \(n\)th roots of unity, acting on \(X\) by rotations about the origin. Then \( \int_G f(g) \, dg = \frac{1}{n}\sum_{r=0}^{n-1} f(\omega^r) \), and \( d\bar{z} \, dz \) is an invariant measure on \(X\).
\item Let \( X = \reals \times S^1 \), and let \( G \cong [0,2\pi) \) with addition modulo \(2\pi\) as the group operation. \(G\) acts on \(X\) by ``rotation of the second factor'', in that \( \theta.(x,e^{iy}) = (x,e^{i(y+\theta)}) \). The invariant measure on \(G\) is simply \( \int_G f(g) \, dg = \int_0^{2\pi} f(\theta) \, d\theta \), and \( dx \times dy\) is invariant on \(X\).
\item Let \( X=\reals^n \), \( G=SO(\Dim) \), acting on \(X\) in the usual way. There is a unique Haar (probability) measure on \(G\), and ordinary Lebesgue measure \( dx \) is again invariant.
\end{enumerate}

The first two examples illustrate that we shall not need to take \(G\) to be a Lie group: discrete symmetry groups will also work. The last example is that most commonly considered in this type of work: spherical symmetry.

We take \( \cX \) to be a complete space of functions on \(X\), and have the crucial:

\begin{definition}[Ground state]
  Let \(E : \cX \to [-\infty,\infty] \). A function \(u\) is said to be an \emph{unconstrained ground state} of \(E\) if
  \begin{equation*}
    E[u] = \inf_{v \in \cX} E[v] > -\infty.
  \end{equation*}
  Similarly, if we let \( \cY\) be a norm-closed subset of \(\cX\) and then \(u\) is said to be a \emph{constrained ground state} of \(E\) if
  \begin{equation*}
    E[u] = \inf_{v \in \cY} E[v] > -\infty.
  \end{equation*}
\end{definition}

In fact, while the distinction between these is frequently important, the difference between these two is not actually of great significance in the cases we shall consider, at least when constructing the proofs.

\section{Energy Functionals Considered}
\label{sec:lagr-cons}

We now give in detail the expressions of energy functionals that shall be susceptible to our proof. We may have the following:
\begin{enumerate}
\item The \emph{kinetic term},
  \begin{equation}
    T[u] := \int_X \abs{\nabla u}^2 = \norm{ \abs{\nabla u} }_2^2.
  \end{equation}
\item The \emph{external potential term},
  \begin{equation}
    P[\abs{u}^2] := \int_X V \abs{u}^2,
  \end{equation}
  where \( V \) is a given fixed function (the \emph{potential}), which for the purposes of this paper is taken to be not equivalent to the zero function.
\item The \emph{self-potential term},
  \begin{equation}
    Q[\abs{u}^2] := \int_X \int_X \abs{u(x)}^2 h(x,y) \abs{u(y)}^2 \, dx \, dy,
  \end{equation}
  where \(h\) is a positive-definite function (that is, \(Q(f) \geq 0\) for all \(f\) such that the integral converges).
\end{enumerate}

In the last sections we shall also look at several other special cases, including a relativistic kinetic term on \( \reals^n\), and a nonlinear term.

\begin{remark}
  The power \(2\) that appears in each of these terms may be replaced by \(1 \leq p < \infty \), but we give our main results using \(2\) since that is the case normally considered in quantum mechanics.
\end{remark}

The constraint we consider is that the total mass of the function \(u\) is known,
\begin{equation}
  N[u] = \int_X \abs{u}^2 = N;
\end{equation}
this of course recalls the idea that \(N=1\) for a wavefunction.\footnote{The constant \(n\) is unimportant, provided it is positive: it can be set to \(1\) by rescaling \(u\).}

We shall prove several theorems of the following form:

\begin{theorem}[Ground States of Symmetric Potentials are Symmetric]
  Let \( X,G,dg,dx \) be as in \S~\ref{sec:definitions}, \( V(g.x) = V(x) \) and \( h(g.x,g.y) = h(x,y) \) for every \(g \in G\), and
  \begin{equation}
    E[u] := T[u] + P[\abs{u}^2] + b Q[\abs{u}^2]
  \end{equation}
  Suppose one of the following holds:
  \begin{enumerate}
  \item \( X=\reals^{\Dim} \), \(G=SO(\Dim)\), \( b \geq 0\)
  \item \(b>0\).
  \end{enumerate}
Then if \(u\) is a (constrained or unconstrained) ground state of \(E\), it must satisfy \( u(g.x)=u(x) \) for every \(g \in G\).
\end{theorem}

To do this, we shall show that averaging \(u\) over \(g\) cannot increase each of the terms in \(E\), and in the circumstances given, \(T\) and \(Q\) actually decrease. Again, we make no statement about existence: this is only a necessary condition.

The slightly strange set of conditions are explained when we discuss what is required to force \(T\) to decrease.

\begin{remark}
  The compactness assumption is essential: if we consider a \(V\) on the plane that is invariant in one coordinate, then with \(G\) the group of translations in this coordinate, the \(G\)-invariant functions will not be integrable.
\end{remark}

\section{Increasing Potentials: The Symmetric Decreasing Rearrangement}
\label{sec:incr-potent-symm}

We shall first consider a case that can be handled by the usual technique used to produce symmetric solutions, in order to understand what is required. 

\begin{theorem}
  Let \(X =\reals^{\Dim}\), and assume that \( h(x,y)=h(x-y) \). Let \( V,h \in L^1_{\text{loc}}(\reals^{\Dim}) \) be spherically symmetric and increasing, and let \( V \) be strictly increasing. Then if \(  E[u] = T[u] + P[\abs{u}^2] + b Q[\abs{u}^2] \) (\(b\geq 0\)) has a normalised ground state \(u\) (i.e. \(u \in H^1(\Omega)\) and \( \norm{u}_2=1 \)), \(u\) is spherically symmetric and decreasing.
\end{theorem}

The standard method to deal with such cases is to introduce the \emph{symmetric decreasing rearrangement of \(u\)}, defined using the layer-cake decomposition
\begin{equation}
  \abs{u(x)} = \int_0^{\infty} \mu{\{ x: \abs{u(x)}<t \}} \, dt
\end{equation}
and that the rearrangement of a set \(A \subset \reals^{\Dim} \) is \(A^*\), the ball of volume \(\mu(A)\) centred at \(0\). Then
\begin{equation}
  u^*(x) = \int_0^{\infty} \mu{(\{ x: \abs{u(x)}<t \}^*)} \, dt.
\end{equation}
is spherically symmetric and decreasing.

It is easy to see that this operation preserves the norm:
\begin{equation}
  \norm{u}_p = \int_0^{\infty} t^{p-1} \mu{\{ x: \abs{u(x)}<t \}} \, dt = \int_0^{\infty} t^{p-1} \mu{\{ x: \abs{u^*(x)}<t \}} \, dt  = \norm{u^*}_p.
\end{equation}

To prove the theorem, we use several classical inequalities involving \( u^*\):
\begin{enumerate}
\item The \emph{P\'olya--Szeg\H{o} inequality}:\footnote{\cite{Polya:1945dz}, \cite{lieb2001analysis}, Lemma 7.17, equality cases by \cite{Brothers:1988qf}}
  \begin{equation}
    \label{eq:polya-szego}
    \norm{\abs{\nabla u}}_p \geq \norm{\abs{\nabla u^*}}_p,
  \end{equation}
  with equality for \(p>1\) if and only if \(u(x)=u^*(x-x_0)\) almost everywhere.
\item The \emph{Hardy--Littlewood inequality}\footnote{Strangely, although this was first proven in \cite{HLP:1952}, it is normally only attributed to Hardy and Littlewood.}: if \( u,v\) are integrable and nonnegative, then
  \begin{equation}
    \int_{\reals^{\Dim}} uv \leq \int_{\reals^{\Dim}} u^* v^*,
  \end{equation}
  and its simpler cousin: if \(V\) is strictly increasing, then
  \begin{equation}
    \int_{\reals^{\Dim}} V \abs{u}^2 \geq \int_{\reals^{\Dim}} V (u^*)^2,
  \end{equation}
  with equality if and only if \( \abs{u}=u^* \) almost everywhere.
\item The strict form of \emph{Riesz's inequality}\footnote{\cite{lieb2001analysis}, Theorem 3.9}: if \(u,v\) are nonnegative, and \( h:\reals^{\Dim} \to \reals \) is spherically symmetric and increasing, then
  \begin{equation}
    \int_{\reals^{\Dim}} \int_{\reals^{\Dim}} u(x) h(x-y) v(y) \, dx \, dy \geq \int_{\reals^{\Dim}} \int_{\reals^{\Dim}} u^*(x) h(x-y) v^*(y) \, dx \, dy;
  \end{equation}
  if  \( h \) is strictly increasing, equality occurs if and only if \( u(x) = u^*(x-x_0) \) and \( v(x)=v^*(x-x_0) \) almost everywhere, for some \(x_0 \in \reals^{\Dim} \).
\end{enumerate}

\begin{proof}[Proof of Theorem]
  Let \(u\) be a minimiser of \(E\). Because \( N[u]=\norm{u}_2^2=\norm{u^*}_2^2 \), \(u^*\) is also in the set over which \(E\) is minimised. Since \(u\) is a minimiser, the three integrals in \(E\) exist, and by the above three inequalities,
  \begin{align*}
    T[u] &\geq T[u^*], \\
    P[\abs{u}^2] &\geq P[\abs{u^*}^2], \\
    Q[\abs{u}^2] &\geq Q[\abs{u^*}^2].
  \end{align*}
  Equality occurs in the second of these if and only if \(u=u^*\), and hence if \(u\) is not symmetric decreasing, the energy can be decreased further by replacing \(u\) by \(u^*\), contradicting the minimality of \(u\). Therefore \(u\) must by symmetric decreasing to be a minimiser.
\end{proof}

\begin{remark}
  Strictly increasing is necessary, at least for this type of proof: one can imagine taking \( V \) to be some kind of bump function, so that contributions of \(u\), also a bump function, could be unchanged by a small translation.
\end{remark}

The proof of our result for potentials that are not strictly increasing will also proceed along these lines, but with replacements for the inequalities that have been used.

\section{The Orbital Mean and Its Basic Properties}
\label{sec:symm-its-basic}

We introduce the following notion:\footnote{This is a generalisation of a standard averaging over circles used in complex analysis: see, for example \cite{Littlewood:1926kx} and \cite{Hardy:1932uq}, whose notation is adapted.}

\begin{definition}
  Let \( X \) be a space, \( u \in L^p(X) \), and let \( G\) be a compact group acting on \( X \). The \emph{\((G,p)\)-orbital mean} of \(u\) is the function given by
  \begin{equation}
    M_p(u)(x) = \left( \int_G \abs{u(g.x)}^p \, dg \right)^{1/p},
  \end{equation}
  where \(dg\) is the invariant probability measure on \(G\).
\end{definition}

Obviously \(M_p\) is nonnegative and \(G\)-invariant. We give a brief sketch of some more simple properties of the \(M_p\), to be borne in mind in the sequel.

\begin{lemma}
  \( M_p(u) \) is increasing as a function of \(p\), and moreover, \( M_p(u) \) is also log-convex, in that if \( p \leq q \leq r \), and \( \theta \) is such that
  \begin{equation*}
    \frac{1}{q} = \frac{\theta}{p}+\frac{1-\theta}{r},
  \end{equation*}
 then
  \begin{equation}
    M_q(u) \leq (M_p(u))^{\theta} (M_r(u))^{1-\theta}
  \end{equation}
\end{lemma}

\begin{proof}
  Both of these results have standard proofs using H\"older's inequality, the additional interest perhaps stemming from the variables remaining present (although we shall elide them for the proof). We write \( u_g = u \circ g \). We then have, using H\"older's inequality,
  \begin{align*}
    M_q(u) &=  \left( \int_G \abs{u_g}^q \right)^{1/q} \\
           &= \left( \int_G \abs{u_g}^{q\theta} \abs{u_g}^{q(1-\theta)} \right)^{1/q} \\
           &\leq \left( \int_G \abs{u_g}^{q\theta (p/\theta q)} \right)^{ (q\theta/p) 1/q} \left( \int_G \abs{u_g}^{q(1-\theta) (r/(1-\theta q))} \right)^{(q(1-\theta)/r) 1/q} \\
           &= \left( \int_G \abs{u_g}^{p} \right)^{\theta/p} \left( \int_G \abs{u_g}^{r} \right)^{(1-\theta)/r} \\
           &=  (M_p(u))^{\theta} (M_r(u))^{1-\theta}.
  \end{align*}
  Likewise, the first part follows from H\"older's inequality: if \(p<q\), then
  \begin{equation*}
    M_p(u) \leq \left( \int_G \abs{u_g}^q \right)^{1/q} \left( \int_G 1^{q/(q-p)} \right)^{(q-p)/pq} = M_q(u). \qedhere
  \end{equation*}
\end{proof}

Equality in each case occurs if and only if \( u_g \) is constant over \(G\), i.e., \(u\) is \(G\)-invariant.

\begin{lemma}[Norm-preservation]
  \begin{equation}
    \norm{u}_p = \norm{M_p(u)}_p.
  \end{equation}
\end{lemma}

\begin{proof}
  We have
  \begin{align*}
    \norm{M_p(u)}_p^p &= \int_{X} \abs{M_p(u)(x)}^p dx \\
                       &= \int_{X} \left( \int_G \abs{u(gx)}^p dg \right)  dx \\
                      &= \int_G \left( \int_X \abs{u(gx)}^p \right) dg \\
                      &= \int_G \norm{u}_p^p dg \\
                       &= \norm{u}_p^p,
  \end{align*}
  where Tonelli's theorem is used for the third equality.
\end{proof}

So far, the orbital mean possesses several of the useful properties of the symmetric decreasing rearrangement. This continues to be the case when we look for equivalents of the inequalities described in the previous section.

\section{The Kinetic Energy}
\label{sec:kinetic-energy}

Here, we would clearly like an analogue of the P\'olya--Szeg\H{o} inequality \eqref{eq:polya-szego}.

\subsection{A Convexity Inequality for the Gradient}
\label{sec:conv-ineq-grad}

We make use of a new, general inequality, which shows that interchanging norms with gradients can only increase the function:

\begin{theorem}
\label{thm:pwconvexgrad}
  Suppose that \(Y\) is a measurable space, \( \Omega \subseteq X \), and for all \(x \in \Omega\), \( f(x,y),\nabla_{x} f(x,y) \in L^{p}(Y) \), and \( \norm{f_x}_p, \norm{\abs{\nabla_x f_x}}_p \in L^{1,P}(\Omega) \). Then \(\norm{f_{x}}_{p} \in W^{1,P}(\Omega) \), and for a.e. \(x \in \Omega\),
  \begin{equation}
    \abs{\nabla_{x} \left( \int_{Y} \abs{f(x,y)}^{p} \, dy \right)^{1/p}} \leq \left(\int_{Y} \abs{\nabla_{x} f(x,y)}^{p} \, dy \right)^{1/p},
  \end{equation}
  i.e., if we write \( f(x,y)=f_{x}(y) \),
  \begin{equation}
    \abs{\nabla_{x} \norm{f_{x}}_{p}} \leq \norm{\abs{\nabla_{x} f_{x}}}_{p},
  \end{equation}
  where norms are taken over \(y\). Equality occurs only if \( \abs{\nabla_x \abs{f(x,y)}^{2-p}} \) (or \( \abs{\nabla_x \log{\abs{f(x,y)}}}  \) if \(p=2\)) is independent of \(y\).
\end{theorem}

\begin{proof}
  From the argument in \cite{lieb2001analysis}, Theorem 6.17, we have
  \begin{equation*}
    (\nabla_x \abs{f})(x,y) =
    \begin{cases}
      \Re\left( \frac{\bar{f}}{\abs{f}} \nabla_x f \right)(x,y) & f \neq 0 \\
      0 & f=0
    \end{cases}.
  \end{equation*}
  It then follows that
  \begin{align*}
    \abs{\nabla_{x} \norm{f_{x}}_{p}} &= \abs{ \norm{f_{x}}_{p}^{1-p} \int \abs{f_{x}}^{p-2} \Re(\bar{f}_{x} \nabla_{x} f_{x}) } \\
                                      &\leq \norm{f_{x}}_{p}^{1-p} \int \abs{f_{x}}^{p-1} \abs{\nabla_{x}f_{x}} \\
                                      &\leq \norm{f_{x}}_{p}^{1-p} \norm{\abs{f_{x}}^{p-1}}_{q} \norm{\abs{\nabla_{x}f_{x}}}_{p} \\
                                      &= \norm{f_{x}}_{p}^{1-p} \norm{f_{x}}_{p}^{p-1} \norm{\abs{\nabla_{x}f_{x}}}_{p} \\
                                      &= \norm{\abs{\nabla_{x}f_{x}}}_{p},
  \end{align*}
  using the chain rule to find the derivative of the \(p\)-norm, applying H\"older's inequality and using that \(1/p+1/q=1\). The equality condition follows from that for H\"older's inequality.
\end{proof}

\begin{remark}
  Exactly the same proof will produce a generalisation of the \emph{diamagnetic inequality},
  \begin{equation}
    \label{eq:29}
    \abs{\nabla\abs{f(x)}} \leq \abs{(\nabla+iA)f(x)},
  \end{equation}
  for any \( f \) in the magnetic space \( H^1_A(\Omega) \), namely
  \begin{equation}
    \label{eq:30}
    \abs{\nabla \norm{f}_{p}} \leq \norm{\abs{(\nabla_x+iA)f_x}}_{p}
  \end{equation}
\end{remark}

In the result below we write for brevity \( \norm{f(x,y)}_{p[2]} = \left( \int_Y \norm{f(x,y)}^p \, dy \right)^{1/p} \) and similarly for \(x\).

\begin{corollary}
  Suppose \( \norm{f(x,y)}_{p[2]} \) exists for almost all \(x\), and \( \norm{\nabla_{x} f(x,y)}_{p[2]} \in L^{P}(\Omega,dx) \). Then \( \abs{\nabla_{x} \norm{f(x,y)}_{p[2]}}^{p} \in L^{P}(\Omega,dx) \), and
  \begin{equation}
    \label{eq:LLgencongrad}
    \left(\int_{\Omega} \abs{\nabla_{x} \left( \int_{Y} \abs{f(x,y)}^{p} \, dy \right)^{1/p} }^{P} \, dx \right)^{1/P} \leq \left( \int_{\Omega} \left( \int_{Y} \abs{\nabla_{x} f(x,y)}^{p} \, dy \right)^{P/p} \, dx \right)^{1/P},
  \end{equation}
  or
  \begin{equation}
    \norm{\abs{\nabla_{x}\norm{f(x,y)}_{p[2]}}}_{P} \leq \norm{\norm{\abs{\nabla_{x} f(x,y)}}_{p[2]}}_{P}.
  \end{equation}
\end{corollary}

\begin{proof}
  Take the \(P\)-norm of the previous result.
\end{proof}

\begin{remark}
  This is a considerable generalisation of the inequality
  \begin{equation*}
    \int_{\reals^{\Dim}} \abs{\nabla\sqrt{f^{2}+g^{2}}}^2 (x) \, dx \leq \int_{\reals^{\Dim}} \left( \abs{\nabla f}^{2}(x) + \abs{\nabla g}^{2}(x) \right) \, dx ,
  \end{equation*}
  Theorem 7.8 of \cite{lieb2001analysis}, p.~177, which follows from \eqref{eq:LLgencongrad} by taking \(Y\) to be a two-point space, \(\Omega=\reals^{\Dim}\), and \(p=P=2\).
\end{remark}

In the case of most interest to us, \(p=2\), we can give a more explicit proof of \autoref{thm:pwconvexgrad}: in a similar way to the proof of \cite{lieb2001analysis}, Theorem 7.8, we can use the identity
\begin{align*}
  &\frac{1}{4}\iint \left(\abs{\bar{u}(y)v(z) - \bar{u}(z) v(y)  }^2 + \abs{\bar{u}(y)v(z) - u(z) \bar{v}(y)  }^2 \right) dy \, dz \\
&= \norm{u}_2^2 \norm{v}_2^2 - \iint \Re( \bar{u} v )(y) \cdot \Re( \bar{u} v )(z) \, dy \, dz \\
&= \norm{u}_2^2 \norm{v}_2^2 - \abs{ \int  \Re( \bar{u} v )(y) \, dy }^2,
\end{align*}
with \(u=f_x\) and \( v=\nabla_x f_x \), to obtain
\begin{align*}
  &\frac{1}{4}\iint \left(\abs{\bar{f}_x(y)\nabla_x f_x(z) - \bar{f}_x(z) \nabla_x f_x(y)  }^2 + \abs{\bar{f}_x(y)\nabla_x f_x(z) - f_x(z) \nabla_x \bar{f}_x(y)  }^2 \right) dy \, dz \\
  &=\norm{f_x}_2^2 \norm{\abs{\nabla_x f_x}}_2^2 - \abs{ \int  \Re( \bar{f}_x \nabla_x f_x )(y) \, dy }^2 \\
&= \norm{f_x}_2^2 \norm{\abs{\nabla_x f_x}}_2^2 - \norm{f_x}_2^2 \abs{\nabla_x \norm{f_x}_2}^2.
\end{align*}
The left-hand side is clearly nonnegative, and the result follows. We also see that for equality, we need \( \bar{f}_x(y)\nabla_x f_x(z)  \) to be real and equal to \( f_x(z) \nabla_x \bar{f}_x(y) \) a.e.. Assuming that \(f\) is never zero, we find by a similar method to the latter part of the proof in  \cite{lieb2001analysis} that if equality holds, we can factorise \(f\) as \( f(x,y) = k(y)F(x) \), where \( F(x) \) is real and positive and \(k \neq 0\).: observe that
\begin{equation*}
  (f_x(z))^2\nabla_x(\bar{f}_x(y)/f_x(z)) = f_x(z) \nabla_x \bar{f}_x(y) - \bar{f}_x(y)\nabla_x f_x(z),
\end{equation*}
and writing \(f\) in polar form, \(f(x,y) = e^{r_x(y)+i\theta_x(y)}\), we see that for the left-hand side to vanish, neither \( r_x(y)-r_x(z) \) nor \( \theta_x(y)+\theta_x(z) \) can depend on \(x\) for a.e. \(y,z\). Hence we can write \( f(x,y)/f(x,z) = K(y,z) \) for some nonvanishing \(K\); choosing a specific \(x\), we see that we must be able to write \(K(y,z)=k(z)/k(y)\) for some nonvanishing function \(k\), and then \( f(x,y) = k(y) F(x) \) for some positive function \(F\).

If \(Y=G\) is a group and \( f(x,g)\) is defined by a group action, \( f(x,g) = u(g.x) \), then we can show that \(k(g)\) is a character for the group \(G\): we have \( f(x,g) = u(g.x) = f(g.x,e) \). We may choose \(k(e)=1\) since the scaling of \(k\) and \(F\) may be chosen freely. Then
\begin{equation*}
  F(g.x) = u(g.x) = f(g.x,e) = f(x,g) = k(g) F(x),
\end{equation*}
and then
\begin{equation*}
  k(gg')F(x) = F(gg'.x) = F(gg'.x) = F(g.(g'.x)) = k(g)F(g'.x) = k(g)k(g') F(x),
\end{equation*}
and since \(F \neq 0\), \( k(gg') = k(g)k(g') \) for all \(g,g' \in G\).

If we also impose \( \norm{u_g}_2 = \norm{u}_2 \), we must have \( \abs{k(g)}=1 \), so \(k\) is a \emph{unitary} character.

\subsection{The Orbital Mean and the Kinetic Energy}
\label{sec:symm-kinet-energy}

The convexity inequality leads us almost directly to the inequality we require:

\begin{corollary}[Symmetrisation of the kinetic energy]
  \label{thm:symkinen}
  Let \(u \in H^{1}(\reals^{\Dim}) \). Then
  \begin{equation}
    \norm{\nabla M_2(u)}_{2} \leq \norm{\nabla u}_{2},
  \end{equation}
  with equality if and only if \(u\) is equivalent to a nonnegative radial function.
\end{corollary}

\begin{proof}
  Taking \(p=P=2\), \( x=r \), \(y=g\), \(f(x,y)=u(x y.n)\) with \(n\) a fixed unit vector, \((\Omega,dx)=(\reals^{+},r^{\Dim-1} dr)\), \( (Y,dy)=(O(\Dim), S_{\Dim-1} dg) \) in the inequality \eqref{eq:LLgencongrad}, the right-hand side becomes
  \begin{equation*}
    \left( \int_{0}^{\infty} S_{\Dim-1}\int_{O(\Dim)} \abs{\nabla_{r} u(rg.n)}^{2} \, dg \, dx \right)^{1/2} = \norm{\nabla_r u}_2 \leq \norm{\nabla u}_{2};
  \end{equation*}
  by Tonelli's theorem, with equality only if \(u\) is equivalent to a nonnegative radial function (and hence \(u(x)=M_2(u)(x)\) a.e.). For the left-hand side, we have
  \begin{align*}
    \left(\int_{0}^{\infty} \abs{\nabla_{r} \left( S_{\Dim-1}\int_{O(\Dim)} \abs{u(rg.n)}^{2} dg \right)^{1/2} }^{2} r^{\Dim-1} \, dr \right)^{1/2} &= \left(\int_{0}^{\infty} S_{\Dim-1} \abs{\nabla_{r} M_2(u)(rn) }^{2} r^{\Dim-1} \, dr \right)^{1/2} \\
                                                                                                                                                   &= \left(\int_{0}^{\infty} S_{\Dim-1} \abs{\nabla M_2(u)(rn) }^{2} r^{\Dim-1} \, dr \right)^{1/2} \\
                                                                                                                                        &= \norm{\nabla M_2(u)}_{2},
	\end{align*}
	since \(M_2(u)(rn)\) is independent of \(n\).
\end{proof}

\begin{remark}
  The nature of this proof required that we were able to essentially decompose \(\reals^{\Dim}\) into subspaces perpendicular and parallel to the action of \(SO(\Dim)\), i.e.
  \begin{equation*}
    \reals^{\Dim} \cong \reals^+ \times S^{\Dim-1} = \reals^+ \times \orb_G(n),
  \end{equation*}
  where in the second expression, \(SO(\Dim)\) acts only on the second factor, and in the last, \(n \in S^{\Dim-1}\) is a fixed vector. The same principle applies to the third example we gave in the second section, in that \(G\) acts only on \(y\); hence, the same proof will apply to any space which may be decomposed as \( X \cong Y \times Z \) with \(G\) acting only on the factor \(Z\), and \( Z=\orb_G(z) \) for some \(z \in Z\). We then obtain a corresponding generalisation of the results below that require that \(T\) actually decreases, but we have omitted this from the theorem statements for the sake of simplicity. A better understanding of the equality cases in \autoref{thm:pwconvexgrad} may allow this to be relaxed.
\end{remark}

\section{The Potential Energy}
\label{sec:potential-energy}

The results for the potential energy are obtained quite differently. The external potential term is straightforward, but introduces an important idea:

\begin{result}
  If \( V \) and \(dx\) are \(G\)-invariant, then
  \begin{equation}
    \int_{X} V(x) \abs{u(x)}^2 = \int_{X} V(x) (M_2(u)(x))^2 \, dx.
  \end{equation}
\end{result}

\begin{proof}
  Changing variables, we have
  \begin{equation*}
    \int_{X} V(x) \abs{u(x)}^2 dx = \int_X V(g.x) \abs{u(g.x)}^2 dx = \int_X V(x) \abs{u(g.x)}^2 dx,
  \end{equation*}
  since \(V\) is \(G\)-invariant. Integrating over \(G\), the left-hand side remains the same, but on the right we have
  \begin{align*}
    \int_G \int_{X} V(x) \abs{u(x)}^2 dx \, dg
    &= \int_G \int_X V(x) \abs{u(g.x)}^2 dx \, dg \\
    &= \int_X V(x) \int_G \abs{u(g.x)}^2 dg \, dx \\
    &= \int_X V(x)  (M_2(u)(x))^2 \, dx,
  \end{align*}
  as required.
\end{proof}

The key idea here is that the invariance of the measure \(dx\), which enables us to change variables and interchange the order of integration; this will be essential to the argument for the other inequalities.

For now, however, we shall tack, and look at the self-potential term, which requires a different type of rearrangement. We first need a result for convolutions that, phrased properly, looks simple:

\begin{lemma}
  If \(G\) acts linearly on \( X\), \( dx\) is \(G\)-invariant and \( U,h \) are \(G\)-invariant then
  \begin{equation*}
    W(y) = (U * h)(y) = \int_X U(x) h(x-y) \, dx
  \end{equation*}
  is also \( G \)-invariant.
\end{lemma}
\begin{proof}
  We have
  \begin{align*}
    W(g.y) &= \int_X U(x) h(x-g.y) \, dx \\
          &= \int_X U(g.x') h(g.x'-g.y) \, dx' \\
          &= \int_X U(x') h(x'-y) \, dx' \\
          &= W(y),
  \end{align*}
  applying the substitution \( x'=g.x \) in the second equality.
\end{proof}

We say \(h\) is \emph{positive-definite} if
\begin{equation}
  \label{eq:+vedef}
  \int_X \int_X u(x) h(x-y) u(y) \, dx \, dy > 0
\end{equation}
for any nonzero \(u\) so that the integral converges.

\begin{theorem}[Linear symmetraveraging of convolution, I]
  \label{thm:linearconvI}
  Let \( h\) be a \(G\)-invariant, positive-definite function, with the integral in \eqref{eq:+vedef} convergent on a complete linear space of integrable functions \( \mathcal{I} \). Then if \( u \in \mathcal{I} \), we have
  \begin{equation}
    \label{eq:6}
    \int_X \int_X u(x) h(x-y) u(y) \, dx \, dy \geq \int_X \int_X U(x) h(x-y) U(y) \, dx \, dy,
  \end{equation}
  where \( U(x) = \int_G u(g.x) \, dg \in \mathcal{I} \) (i.e. \(M_1(u)\) without the absolute value); equality holds if and only if \( u\) is \(G\)-invariant a.e..
\end{theorem}

\begin{proof}
  We have the equality \( u = U + (u-U) \). If we write \( H(u,v) \) for the left-hand side of the inequality in the theorem, we have
  \begin{align}
    \label{eq:linearconvdecomp}
    H(u,v) &= H\big(U+(u-U),U+(u-U)\big) \nonumber\\
           &= H(U,U) + 2H(U,u-U) + H(u-U,u-U).
  \end{align}
  The first term is the right-hand side of the inequality. The last term is nonnegative by the positive-definiteness of \(h\), and in particular is zero unless \(u=U\) a.e., i.e. \(u\) is \(G\)-invariant a.e..
  
  It remains to deal with the middle term. The previous lemma shows that \( V = U * h \) is \(G\)-invariant, so setting \(y=g.y'\),
  \begin{align*}
    \int_X W(y) (u(y)-U(y)) \, dy &= \int_X \left(\int_G W(g^{-1}.y) \, dg \right) \big( u(y)-U(y) \big)  \, dy \\
                                  &= \int_G \left( \int_X W(g^{-1}.y) \big( u(y)-U(y) \big) \, dy \right) dg \\
                                  &= \int_G \left( \int_X W(y') \big( u(g.y')-U(g.y') \big) \, dy' \right) dg \\
                                  &= \int_X W(y') \left( \int_G \big( u(g.y')-U(g.y') \big) \, dg \right) dy' \\
                                  &= \int_X W(y') \left( \int_G u(g.y') \, dg -U(y') \right) dy' \\
                                  &= 0
  \end{align*}
  by the definition of \(U\). The result follows.
\end{proof}

\begin{remark}
  We can generalise this theorem to \(h=h(x,y)=h(g.x,g.y)\) in exactly the same way; this shall prove useful in applications. Likewise, the proof also applies if we replace \( u \) by \( u+\rho \), where \(\rho\) is any \(G\)-invariant function with finite integrals against \(h\).
  
  Additionally, we note that positive-definiteness is equivalent to being the Fourier transform of a positive measure, although this does not determine the space on which the function is positive-definite directly.
\end{remark}

\section{The Ground State is Symmetric}
\label{sec:ground-state-symm}

In this section, we collect the previous results together to prove

\begin{theorem}[Ground States of Spherically Symmetric Potentials in \(\reals^{\Dim}\) are Symmetric]
  Let \( X=\reals^{\Dim} \), \( G=SO(\Dim) \) be as in \S~\ref{sec:definitions}, \( V(g.x) = V(x) \) for every \(g \in G\), and
  \begin{equation}
    E[u] = T[u] + P[\abs{u}^2] + b Q[\abs{u}^2]
  \end{equation}
with \(b \geq 0 \). Then if \(u\) is a (constrained or unconstrained) ground state of \(E\), it must satisfy \( u(g.x)=u(x) \) for every \(g \in G\).
\end{theorem}

\begin{proof}
  Let \(u\) be the minimiser. We have \( N=\norm{u}_2^2 = \norm{M_2(u)}_2^2 \), so \(M_2(u)\) is also in the space over which we minimise. Moreover, the previous two sections give us the inequalities
  \begin{align*}
    T[u] &\geq T[M_2(u)], \\
    P[u] &= P[M_2(u)], \\
    Q[u] &\geq Q[M_2(u)].
  \end{align*}
  Equality occurs in the first if and only if \(u\) is \(G\)-invariant. Hence if \( u \) is not \( G\)-invariant, \( E[u] > E[M_2(u)] \), contradicting minimality of \(u\).
\end{proof}

\begin{theorem}[Ground States of Symmetric Potentials with Self-Attraction are Symmetric]
\label{thm:Gselfinvaratt}
  Let \( X,G,dg,dx \) be as in \S~\ref{sec:definitions}, \( V(g.x) = V(x) \) for every \(g \in G\), and
  \begin{equation}
    E[u] = T[u] + P[\abs{u}^2] + b Q[\abs{u}^2]
  \end{equation}
  Suppose \(b>0\). Then if \(u\) is a (constrained or unconstrained) ground state of \(E\), it must satisfy \( u(g.x)=u(x) \) for every \(g \in G\).
\end{theorem}

The proof of this is much the same as the previous theorem, but uses the third term:

\begin{proof}
  Let \(u\) be the minimiser. We have \( N=\norm{u}_2^2 = \norm{M_2(u)}_2^2 \), so \(M_2(u)\) is also in the space over which we minimise. Moreover, the previous two sections give us the inequalities
  \begin{align*}
    T[u] &\geq T[M_2(u)], \\
    P[u] &= P[M_2(u)], \\
    Q[u] &\geq Q[M_2(u)].
  \end{align*}
  Equality occurs in the third if and only if \(u\) is \(G\)-invariant. Hence if \( u \) is not \( G\)-invariant, \( E[u] > E[M_2(u)] \), contradicting minimality of \(u\).
\end{proof}

This covers attractive self-potential terms, where it is favourable for the field to clump together (this tendency being counteracted by the kinetic energy in the case in point). More difficult are repulsive self-potential energies, where the external potential is vital for a ground state to exist. These issues in themselves do not concern us here; instead, we shall prove that the symmetry result still holds for repulsive self-potential terms, if they satisfy a slightly different criterion for positive-definiteness:

\begin{lemma}
  Suppose that \( h(x,y) \) is positive-definite on the space \( \mathcal{I}_0 \) of integrable functions with integral zero. Let \( u \in L^1 \) be positive and satisfy \( \int_X u = N \). Let \( \rho \) be a positive function with \( \int_X \rho = N \). Then
  \begin{equation}
    H_{\rho}(u,v) = \int_X \int_X (u-\rho)(x) h(x-y) (v-\rho)(y) \, dx \, dy
  \end{equation}
  is a positive-definite on \( \mathcal{I}_N \), the set of functions with integral \( N \).
\end{lemma}

\begin{proof}
  This is clear: since \( \int_X (u-\rho) = 0 \), \(H_{\rho}(u,u) \geq 0 \), with equality if and only if \( u=\rho \), by the positive-definiteness of \( h \) on \( \mathcal{I}_0 \).
\end{proof}

\begin{corollary}[Linear symmetraveraging of convolution, II]
  If \(h,\rho\) are as in the previous lemma and are also \(G\)-invariant, then \( H(u,u) \geq H( U,U ) \), where \( U \) is as in \autoref{thm:linearconvI}.
\end{corollary}

The proof is almost exactly the same as that of \autoref{thm:linearconvI}, save that the decomposition
\begin{align*}
  H_{\rho}(u,u) &= H(u-\rho,u-\rho) \\
                &= H(U-\rho,U-\rho)+2H(u-U,U-\rho)+H(u-U,u-U) \\
                &= H_{\rho}(U,U)+2H(u-U,U-\rho)+H(u-U,u-U)
\end{align*}
is used instead of \eqref{eq:linearconvdecomp}.

With this, we can at last prove
\begin{theorem}
  Let \(V\) be \(G\)-invariant, \(b > 0\), and \(h\) positive-definite on the set of functions with \( \int_X f = 0 \). Then if \(E[\cdot]\) has a ground state \(u\) with \(\norm{u}_2^2 = N \), then \(u\) is \(G\)-invariant.
\end{theorem}

\begin{proof}
  Taking \(\rho\) an integrable, \(G\)-invariant function as above, we have
  \begin{align*}
    E[u] &= T[u]+P[\abs{u}^2] + b Q[u^2] \\
         &= T[u]+P[\abs{u}^2] + b \left( H_{\rho}(\abs{u}^2,\abs{u}^2) - 2H(\rho,\abs{u}^2) + H(\rho,\rho) \right).
  \end{align*}
  Suppose now that \(u\) is a minimiser. We have the equality and inequality
  \begin{align*}
    H(\rho,\abs{u}^2) &= H(\rho,(M_2(u))^2) \\
    H_{\rho}(\abs{u}^2,\abs{u}^2) &\geq H_{\rho}((M_2(u))^2,(M_2(u))^2),
  \end{align*}
  with equality if and only if \( \abs{u} = M_2(u) \) a.e.. These imply that
  \begin{equation*}
    Q[\abs{u}^2] \geq Q[(M_2(u))^2],
  \end{equation*}
  with equality under the same conditions. We also have
  \begin{align*}
    T[u] &\geq T[M_2(u)], \\
    P[\abs{u}^2] &= P[(M_2(u))^2].
  \end{align*}
  It follows that if \( u \) is not \(G\)-invariant, \( E[u] \geq E[M_2(u)] \), which contradicts the minimality of \(u\).
\end{proof}

\section{The Relativistic Kinetic Energy}
\label{sec:relat-kinet-energy}

Alternatively, we may be interested in relativistic scalar equations, where \( T[u] \) is replaced by a relativistic kinetic energy; for Euclidean space \( X=\reals^{\Dim} \), this is given by
\begin{equation}
\label{eq:rKE}
  R[u] := \langle u, (\sqrt{p^2+m^2}-m)u \rangle = \iint \abs{u(x)-u(y)}^2 R_m(x-y) \, dx \, dy,
\end{equation}
where for \( x \in \reals^{\Dim} \),
\begin{equation}
  R_m(x) = \left(\frac{m}{2\pi}\right)^{(\Dim+1)/2} \frac{K_{(\Dim+1)/2}(m\abs{x})}{\abs{x}^{(\Dim+1)/2}},
\end{equation}
where \( K_{\nu}(z) \) is the modified Bessel function asymptotic to \( \sqrt{\pi/(2z)} e^{-z} \) as \( z \to \infty\) in the right half-plane (this may be deduced from \S\S~7.11--2 of \cite{lieb2001analysis}, which does everything but explicitly state the general form).

In order to produce a result similar to that for the ordinary kinetic energy, we shall require equivalents of the other inequalities, which involve several functions. We begin with the equivalent of the Hardy--Littlewood inequality, which is the following:

\begin{theorem}[Orbital Mean of two functions]
  Suppose \( dx \) is \(G\)-invariant, \( 1/p+1/q = 1 \), and \( u \in L^p(X) \), \( v \in L^q(X) \) are nonnegative. Then
  \begin{equation}
    \int_X u(x) v(x) \, dx \leq \int_X M_p(u)(x) M_q(v)(x) \, dx
  \end{equation}
\end{theorem}

\begin{proof}
  Since \(dx\) is \(G\)-invariant, we have
  \begin{align*}
    \int_X u(x) v(x) \, dx &= \int_G \int_X u(g.x) v(g.x) \, dx \, dg \\
                           &= \int_X \left( \int_G u(g.x)v(g.x) \, dg  \right) dx \\
                           &\leq \int_X \left( \int_G \abs{u(g.x)}^p \, dg  \right)^{1/p} \left( \int_G \abs{v(g.x)}^q \, dg  \right)^{1/q} dx \\
                           &= \int_X M_p(u)(x) M_q(v)(x) \, dx,
  \end{align*}
  using H\"older's inequality.
\end{proof}

The cases of equality in such results appear to be in general hard to determine: here, for example, equality occurs in the only inequality used, i.e. H\"older's inequality (for \(1<p<\infty\)) when one of \(u,v\) is identically zero, or
\begin{equation*}
  \abs{v(g.x)} = \lambda(x) \abs{u(g.x)}^{p-1} \iff \abs{v(y)} = \lambda(g^{-1}.y) \abs{u(y)}^{p-1}
\end{equation*}
for some \( \lambda(x)>0\), and almost every \(x\), \(y\) and \(g\). Then the left-hand side is independent of \(g\), so \(\lambda\) is \(G\)-invariant; we have then returned to the case of the invariance of the \(p\)-norm on applying \(M_p\), albeit with a different measure, \(\lambda(x) \, dx\).

The Hardy--Littlewood generalisation is not especially useful to us, but the extension to the convolution integral of three functions, corresponding to Riesz's inequality, is. We have

\begin{theorem}[Orbital Mean of the convolution of three functions]
  Let \(X\) be a vector space, and \( G \) a compact group acting linearly on \( X \); let \( dx \) be an \(G\)-invariant measure on \(X\). Let \( u,v,w \) be positive. Then
  \begin{equation}
    \label{eq:66}
    \int_X \int_X u(x) v(x-y) w(y) \, dx \, dy \leq \int_X \int_X M_p(u)(x) M_q(v)(x-y) M_r(w)(y) \, dx \, dy,
  \end{equation}
  where \( 1/p+1/q+1/r=1 \).
\end{theorem}

\begin{proof}
  The proof is much the same as before: the whole integral is \(G\)-invariant since the measure is, and so
  \begin{align*}
    \int_X \int_X u(x) v(x-y) w(y) \, dx \, dy 
    &= \int_G \left( \int_X \int_X u(g.x) v(g.(x-y)) w(g.y) \, dx \, dy \right) dg \\
    &= \int_X \int_X \left( \int_G u(g.x) v(g.(x-y)) w(g.y) \, dg \right) dx \, dy \\
    &\leq \int_X \int_X M_p(u)(x) M_q(v)(x-y) M_r(w)(y) \, dx \, dy.
  \end{align*}
  by H\"older's inequality for three functions.
\end{proof}

There are obvious generalisations of this to many functions using the more general form of H\"older's inequality.

This is rather more general than we need, however: if we instead suppose that the middle function is already \(G\)-invariant, but not necessarily positive, we obtain something more useful:

\begin{theorem}
  Let \(X\) be a linear space, and \( G \) a compact group acting linearly on \( X \); let \( dx \) be an \(G\)-invariant measure on \(X\). Let \( u,v \geq 0 \), and let \(h\) be \( G \)-invariant. Then
  \begin{equation}
    \label{eq:GRiesz}
    \int_X \int_X u(x) h(x-y) v(y) \, dx \, dy \leq \int_X \int_X M_2(u)(x) h(x-y) M_2(v)(y) \, dx \, dy
  \end{equation}
  and
  \begin{equation}
    \label{eq:GinvarKE}
    \int_X \int_X \abs{u(x)-u(y)}^2 h(x-y) \, dx \, dy \geq \int_X \int_X \abs{M_2(u)(x)-M_2(u)(y)}^2 h(x-y) \, dx \, dy.
  \end{equation}
\end{theorem}

\begin{proof}
  \(h\) is \( G\)-invariant, so
  \begin{equation*}
    h(x-y) = \int_G h(g^{-1}.(x-y)) \, dg.
  \end{equation*}
  Inserting this into the left-hand side of \eqref{eq:GRiesz} and interchanging the order of integration,
  \begin{align*}
    \int_X \int_X u(x) h(x-y) v(y) \, dx \, dy &= \int_X \int_X \left( \int_G h(g^{-1}.(x-y)) \, dg \right) u(x)v(y) \, dx \, dy \\
                                               &= \int_G \int_X \int_X h(g^{-1}.(x-y)) u(x)v(y) \, dx \, dy \, dg.
  \end{align*}
  We can now make the substitution \( x=g.x' \), \( y = g.y' \) to find that
  \begin{align*}
    \int_G \int_X \int_X h(g^{-1}.(x-y)) u(x)v(y) \, dx \, dy \, dg &= \int_G \int_X \int_X h(x'-y') u(g.x')v(g.y') \, dx' \, dy' \, dg \\
    &=  \int_X \int_X h(x'-y') \left( \int_G u(g.x')v(g.y') \, dg \right) dx' \, dy'.
  \end{align*}
  Since \( h \) is nonnegative, we can apply the Cauchy--Schwarz inequality to the \(G\)-integral:
  \begin{align*}
    \int_G u(g.x')v(g.y') \, dg \leq \left( \int_G u(g.x') \, dg \right)^{1/2} \left( \int_G v(g.y') \, dg \right)^{1/2} = M_2(u)(x) M_2(v)(y),
  \end{align*}
  from which the first equation follows. The second inequality is much the same, but we replace Cauchy--Schwarz with the reverse triangle inequality,
  \begin{align*}
    \left(\int_G \abs{u(g.x)-u(g.y)}^2 \, dg \right)^{1/2} 
    &\geq \abs{\left(\int_G \abs{u(g.x)}^2 dg \right)^{1/2}-\left( \int_G \abs{u(g.y)}^2 \, dg \right)^{1/2} } \\
    &= \abs{M_2(u)(x)-M_2(u)(y)}. \qedhere
  \end{align*}
\end{proof}

\begin{remark}
  Equality in the Cauchy--Schwarz inequality\footnote{or equivalently the reverse triangle inequality} is attained when
  \begin{equation}
    \label{eq:34}
    v(g.x) = \lambda(x,y) u(g.y) \implies v(x) = \lambda(g^{-1}.x,g^{-1}.y) u(x),
  \end{equation}
  where \(\lambda\) is a positive, \(G\)-invariant function. By considering the case \(x=y\), it is immediate that \( v(x) = \lambda(g^{-1}.x,g^{-1}.x) u(x) = \lambda(x,x) u(x) \), and so it is sufficient to consider the case \(u=v\), where \( \lambda(x,x)=1 \). If \(X = \reals^{\Dim}\) and \(u\) is radial, \( \lambda \) is purely a function of \( \abs{x},\abs{y} \) and equal to \(1\) for \( \abs{x}=\abs{y} \), but this is by no means the only possibility. An investigation of this is beyond the scope of this work, and in fact unnecessary for our purposes, but may prove fruitful in future.
\end{remark}

We see that this is precisely what we need to prove that \eqref{eq:rKE} cannot increase under the orbital mean. Hence we obtain:

\begin{theorem}
  Let \(V \not\equiv 0\) be \(G\)-symmetric, \(h\) be positive-definite, and \(R\) be as above. Then if \(u\) is a minimiser of
  \begin{equation}
    E_R[u] := R[u] + P[u]+ Q[u],
  \end{equation}
  on the set with \( \norm{u}_2^2=N \) then \(u\) is radial.
\end{theorem}

The proof is exactly the same as that of \autoref{thm:Gselfinvaratt}, with \(T\) replaced by \(R\). The crucial strict inequality is again provided by \(Q\) only.

\section{Other Nonlinearity}
\label{sec:other-nonlinearity}

That is, terms of the form \( \int F(\abs{u}) \), where \( F\) satisfies certain conditions, as, for example, in the nonlinear Schr\"odinger equation. It turns out that we require a specific form of convexity:

\begin{lemma}
  Let \( F(x,u)=f(x,\abs{u}^p) \) be convex in \(\abs{u}^p\) for almost all \(x\), and \( F(g.x,u)=F(x,u) \) for almost every \(g\). Then
  \begin{equation}
    \int_X F(x,u(x)) \, dx \geq \int_X F\left( x, M_p(u)(x) \right) \, dx.
  \end{equation}
  Moreover, if \( f \) is strictly convex in its second argument, equality holds precisely when \(\abs{u} = M_p(u)\).
\end{lemma}

\begin{proof}
  We have
  \begin{equation*}
    \int_X F(x,u(x) ) \, dx = \int_X \int_G F(g.x,u(g.x)) \, dg \, dx = \int_X \int_G F(x,u(g.x)) \, dg \, dx,
  \end{equation*}
  by the same process as the previous proofs. We now apply Jensen's inequality, which in this case says
  \begin{align*}
    \int_G F(x,u(g.x)) \, dg &= \int_G f\left(x,\abs{u(g.x)}^p\right) \, dg \\
                            &\geq f\left( x , \int_G \abs{u(g.x)}^p dg \right) \\
                            &= F\left(x, M_p(u)(x)\right)
  \end{align*}
  for almost every \(x\), and integrating both sides gives the result. The second part follows from the equality case in Jensen's inequality.
\end{proof}

Examples of such functions include \( F(x,u) = a(x)\abs{u}^{p} \) for \(p > 2\), \(a>0\), \(a\) \(G\)-invariant. Exactly the same argument as before gives that the minimisers of energy functionals including this type of nonlinearity are \(G\)-invariant; we may also, if the convexity is strict, replace \(Q\) by this term completely.



\printbibliography

\end{document}